\documentclass[graybox]{svmult}

\usepackage{type1cm}        
\usepackage{makeidx}         
\usepackage{graphicx}        
\usepackage{multicol}        
\usepackage[bottom]{footmisc}
\usepackage{chngpage}
\usepackage{enumerate}
\usepackage{newtxtext}       %
\usepackage{newtxmath}       
\usepackage{hyperref}

\makeindex             

\usepackage{marvosym}
\usepackage{subcaption}

\begin{document}
\title*{A Model for the Outbreak of COVID-19: Vaccine Effectiveness in a Case Study of Italy}
\author{Vasiliki Bitsouni, Nikolaos Gialelis and Ioannis G. Stratis}
\institute{Vasiliki Bitsouni \at Department of Mathematics, National and Kapodistrian University of Athens, Panepistimiopolis, GR-15784 Zographou, Athens, Greece, \email{vbitsouni@math.uoa.gr}\\
\& \\
Applied Mathematics Laboratory, School of Science and Technology, Hellenic Open University, 13-15 Tsamadou str., GR-26222 Patras, Greece.
\and Nikolaos Gialelis \at Department of Mathematics, National and Kapodistrian University of Athens, Panepistimiopolis, GR-15784  Zographou, Athens, Greece, \email{ngialelis@math.uoa.gr} \and Ioannis G. Stratis \at Department of Mathematics, National and Kapodistrian University of Athens, Panepistimiopolis, GR-15784 Zographou, Athens, Greece, \email{istratis@math.uoa.gr}}

\maketitle

\abstract{We present a compartmental mathematical model with demography for the spread of the COVID-19 disease, considering also asymptomatic infectious individuals. We compute the basic reproductive ratio of the model and study the local and global stability for it. We solve the model numerically based on the case of Italy. We propose a vaccination model and we derive threshold conditions for preventing infection spread in the case of imperfect vaccines.}


\section{Introduction}
\label{intro}
In late 2019, the severe acute respiratory syndrome coronavirus 2 (SARS-CoV-2) \cite{gorbalenya2020}, a strain of coronavirus that causes coronavirus disease 2019 (COVID-19), appeared in Wuhan, China, and rapidly spread across the globe. In January 2020 the human-to-human transmission of SARS-CoV-2 was confirmed, during the COVID-19 pandemic, and SARS-CoV-2 was designated a Public Health Emergency of International Concern by the WHO, and in March 2020 the WHO declared it a pandemic. Since December 31, 2019, and as of July 23, 2020, 15,455,848 cases  of COVID-19 (in accordance with the applied case definitions and testing strategies in the affected countries) are confirmed in more than 227 countries and 26 cruise ships \cite{Worldometer}. Currently, there are 5,415,828 active cases and 631,775  deaths.

As of spring 2020, Italy is one of the countries suffering the most with the COVID-19 outbreak. One of the most critical facts about COVID-19, that affected Italy too, is that a significant number of cases, mainly those of young age, has been reported as asymptomatic \cite{Lavezzo2020}, leading to fast spread of the infection. Fortunately, asymptomatic cases have a shorter duration of viral shedding and lower viral load \cite{Yang2020,Li2020}. However, the proportion of asymptomatic cases can range from 4\%-80\% \cite{Heneghan2020} and most of the time they play a key role in infection transmission, therefore it is very important to model both symptomatic and asymptomatic cases. The most crucial element of COVID-19 pandemic is demography. After the end of the quarantine and lock-down there have been many examples of countries where the number of cases increased rapidly.  Although, demography is a very important element in epidemics and especially in the case of COVID-19, only a few models describe the dynamics of infection with SARS-CoV-2  considering demographic terms. To this end, here we develop an SEIAR model, considering both symptomatic and asymptomatic cases, and demographic terms, with a constant influx of susceptible individuals. In this study we only consider horizontal infection transmission\footnote{ Horizontal transmission is transmission by direct contact between infected and susceptible individuals or between disease vectors and susceptible individuals, that are not in a parent-progeny relationship. Vertical transmission is the passage of the disease-causing agent  from mother to baby during the period immediately before and after birth.}.

As the development of a vaccine against SARS-CoV-2 is an urgent demand\footnote{During the days of writing this paper, promising progress has been announced in this direction.}, a study of the epidemiological consequences that an imperfect potential vaccine can have is needed, and to the best of our knowledge, there is currently no relevant study in the literature. Here, we develop a theoretical framework to assess the vaccine effectiveness and the epidemiological consequences of a potential vaccine. Given that the second phase might lead to more asymptomatic cases than the first phase we need to investigate the asymptomatic group for different coverage and efficacy. The model derived in this study provides new insights on the effect of different vaccine coverage and efficacy on infection spread in a population with demographics.

This study is organized as follows. In Section~\ref{modelling}, we develop a new SEIAR model for COVID-19. We derive the $\mathcal{R}_0$ of the model and we study the local and global stability of the steady states. Numerical result with a focus on data fitted to Italy case are presented. In Section \ref{modelling_vaccine}, we extend our model including also a vaccinated group of individuals, we study the global stability of the steady states and we predict the vaccine effectiveness. We conclude in Section \ref{conclusions} with a summary and discussion of future directions.

\section{Modelling Transmission Dynamics}
\label{modelling}
\subsection{An SEIAR Model for COVID-19}
In this Section we model the transmission of SARS-CoV-2 among people causing COVID-19. To derive our model we consider people who have been in contact with an infectious individual, but remain uninfected for a latency period. Moreover, a significant number of individuals being infected remain asymptomatic, due to various factors such as age, health condition etc \cite{Lavezzo2020}. To this end, we classify the total population ($N$) into five subclasses: susceptible ($S$), latent/exposed ($E$), symptomatic infectious ($I$), asymptomatic infectious ($A$), and recovered ($R$) individuals, hence we have that $S+E+I+A+R=N$. We can consider that $N$ is practically constant, since the time-span of the epidemiological phenomenon is relatively short and $N$ is relatively large. We take into account demographic terms and we consider the transmission to be exclusively horizontal. The SEIAR model is governed by the following system of nonlinear ordinary differential equations:
\begin{subequations}
\label{SEIAR}
\begin{align}
\dfrac{\mathrm{d}S}{\mathrm{d}t}&=\mu N-\beta_I SI-\beta_A SA-\mu S,\label{SEIAR;a}\\
\dfrac{\mathrm{d}E}{\mathrm{d}t}&=\beta_I SI+\beta_A SA -\left(k+\mu\right) E,\label{SEIAR;b}\\
\dfrac{\mathrm{d}I}{\mathrm{d}t}&=k \left(1-q\right)E-\left(\gamma_I+\mu\right)I,\label{SEIAR;c}\\
\dfrac{\mathrm{d}A}{\mathrm{d}t}&=k q E-\left(\gamma_A+\mu\right)A,\label{SEIAR;d}\\
\dfrac{\mathrm{d}R}{\mathrm{d}t}&=\gamma_I I+\gamma_A A-\mu  R,\label{SEIAR;e}
\end{align}
\end{subequations}
with initial conditions:
\begin{equation}
\left(S\left(0\right),E\left(0\right),I\left(0\right),A\left(0\right),R\left(0\right)\right)=\left(S_0,E_0,I_0,A_0,R_0\right)\in \left( \mathbb{R}_{0}^{+} \right)^5,
\end{equation}
where $\mu$ is the birth/death rate, $\beta_I$ and $\beta_A$ the transmission rates of $I$ and $A$, respectively\footnote{$\beta_{I,A}\equiv\dfrac{c\,\varpi_{I,A}}{N}$, where $c$ is the average number of close contacts of an individual with other individuals and $\varpi_{I,A}$ is the probability of a contact to be effective in turning an $S$ individual into an $I$ or $A$ one, respectively.}, $k$ the incubation rate, i.e. the rate of latent individuals becoming infectious, $q$ the proportion of asymptomatic infectious individuals, $\gamma_I$ and $\gamma_A$ the recovery rates of infectious  and asymptomatic infectious individuals, respectively. A flow diagram of the model is illustrated in Fig.~\ref{SEIAR_flow_diagram}. A straightforward application of the classical ODE theory yields that the above Cauchy problem is well-posed. 
\begin{figure}[h!]
\centering
\sidecaption[t]
\includegraphics[scale=0.95]{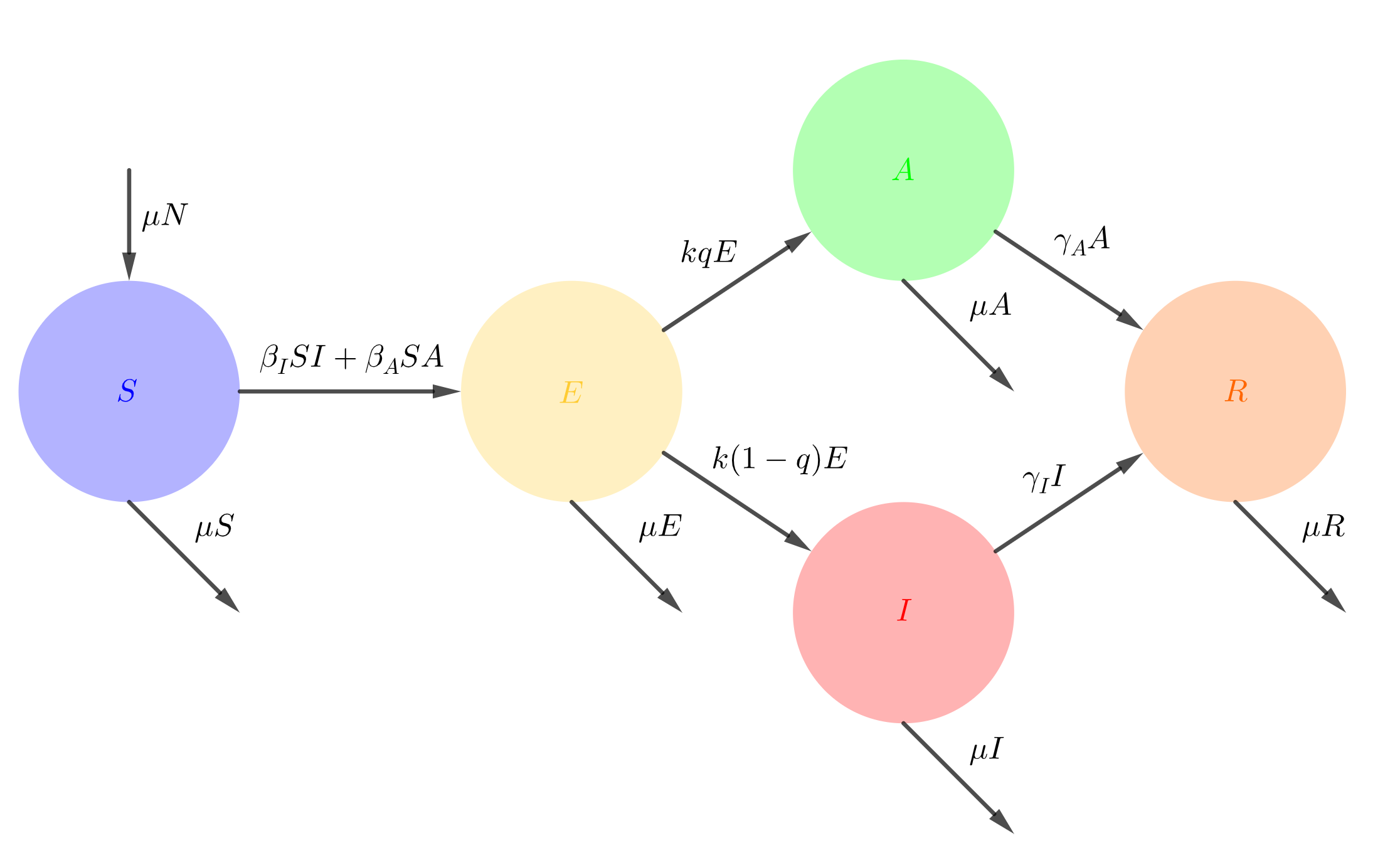}
\caption{Flow diagram of the SEIAR model \eqref{SEIAR}.}
\label{SEIAR_flow_diagram}     
\end{figure}

\subsection{The Basic Reproductive Ratio}
The average number of secondary cases arising from one infection when the entire population
is susceptible is defined as the basic reproductive ratio and denoted by $\mathcal{R}_0$. Being the most important quantity on infectious disease epidemiology, the basic reproductive ratio is a dimensionless quantity, which is  often used to reflect how infectious a disease is. To define $\mathcal{R}_0$ we calculate the next-generation matrix of the system, see, e.g.,  \cite{Diekmann1990,Diekmann2010}. 

First, we linearise model \eqref{SEIAR} around the disease-free steady state, $\left(N,0,0,0,0\right)$,  and we consider the infected states, i.e. $E,I,A$, to obtain the \textit{linearised infection subsystem}
\begin{subequations}
\label{linear_inf_SEIAR}
\begin{align}
\dfrac{\mathrm{d}E}{\mathrm{d}t}&=\beta_I NI+\beta_A NA -\left(k+\mu\right) E,\label{linear_inf_SEIAR;a}\\
\dfrac{\mathrm{d}I}{\mathrm{d}t}&=k \left(1-q\right)E-\left(\gamma_I+\mu\right)I,\label{linear_inf_SEIAR;b}\\
\dfrac{\mathrm{d}A}{\mathrm{d}t}&=k q E-\left(\gamma_A+\mu\right)A.\label{linear_inf_SEIAR;c}
\end{align}
\end{subequations}
Then, we set $\textbf{x}=\left(E,I,A\right)^{\mathrm{tr}}$, so that the system \eqref{linear_inf_SEIAR;a}-\eqref{linear_inf_SEIAR;c}  can be written in the form
\begin{equation*}
   \dfrac{\mathrm{d}\textbf{x}}{\mathrm{d}t}=\left(\textbf{T}+\bf{\Sigma}\right)\textbf{x},
\end{equation*}
where
\begin{equation*}
    \bf{T}=
   \begin{bmatrix}
0 & \beta_I N & \beta_A N\\
0 & 0 & 0\\
0 & 0 & 0\\
\end{bmatrix}
\end{equation*}
is the \textit{transmission matrix}, and
\begin{equation*}
    \bf{\Sigma}=
   \begin{bmatrix}
-\left(k+\mu\right) & 0 & 0\\
k\left(1-q\right) & -\left(\gamma_I+\mu\right) & 0\\
kq & 0 & -\left(\gamma_A+\mu\right)\\
\end{bmatrix}
\end{equation*}
is the \textit{transition matrix}.\\ Then, $\mathcal{R}_0$ is defined as the dominant eigenvalue of matrix $\bf{K}=-\bf{T\Sigma^{-1}}$ as follows:
\begin{align*}
    \bf{K}&=-\bf{T\Sigma^{-1}}=-
   \begin{bmatrix}
0 & \beta_I N & \beta_A N\\
0 & 0 & 0\\
0 & 0 & 0\\
\end{bmatrix}\cdot
\begin{bmatrix}
-\dfrac{1}{\left(k+\mu\right)} & 0 & 0\\
-\dfrac{k\left(1-q\right)}{\left(k+\mu\right)\left(\gamma_I+\mu\right)} & -\dfrac{1}{\left(\gamma_I+\mu\right)}  & 0\\
-\dfrac{kq}{\left(k+\mu\right)\left(\gamma_A+\mu\right)} & 0 & -\dfrac{1}{\left(\gamma_A+\mu\right)}\\
\end{bmatrix}\\&=
\begin{bmatrix}
\dfrac{\beta_I Nk\left(1-q\right)}{\left(k+\mu\right)\left(\gamma_I+\mu\right)}+\dfrac{\beta_A N kq}{\left(k+\mu\right)\left(\gamma_A+\mu\right)} &\; & \dfrac{\beta_I N}{\left(\gamma_I+\mu\right)} &\;& \dfrac{\beta_A N}{\left(\gamma_A+\mu\right)}\\
0 &\;& 0  &\;& 0\\
0 &\;& 0 &\;& 0\\
\end{bmatrix},
\end{align*}
from which we obtain that
\begin{equation}
\mathcal{R}_0=\dfrac{\beta_I Nk\left(1-q\right)}{\left(k+\mu\right)\left(\gamma_I+\mu\right)}+\dfrac{\beta_A N kq}{\left(k+\mu\right)\left(\gamma_A+\mu\right)}.
\label{R0}
\end{equation}

\subsection{Local Stability Analysis of the SEIAR Model}
We proceed with the local stability analysis of the model. The Jacobian matrix of system \eqref{SEIAR} is 
\begin{equation}
\textbf{J}=\begin{bmatrix}
-\beta_I I-\beta_A A-\mu & 0 &  -\beta_I S & -\beta_A S & 0\\
\beta_I I+\beta_A A & -\left(k+\mu\right) &  \beta_I S & \beta_A S  & 0\\
0 & k\left(1-q\right)&  -\left(\gamma_I+\mu\right) & 0 & 0\\
0 & k q &  0 & -\left(\gamma_A+\mu\right) & 0\\
0 & 0 &  \gamma_I & \gamma_A & -\mu\\
\end{bmatrix}.
\label{jacobian}
\end{equation}

\begin{theorem}
If $\mathcal{R}_0<1$, the disease-free steady state, $\left(N,0,0,0,0\right)$, of system \eqref{SEIAR}  is locally stable.
\end{theorem}
\begin{proof}
For the disease-free steady state $\left(N,0,0,0,0\right)$ we obtain a double negative eigenvalue, $\lambda_1=-\mu$, and the characteristic equation of the reduced 3x3 matrix\\
\begin{align*}
\lambda^3+&\biggl(\gamma_I+\gamma_A+k+3\mu\biggr)\lambda^2+\biggl[\left(k+\mu\right)\left(\gamma_I+\gamma_A+2\mu\right)+\left(\gamma_I+\mu\right)\left(\gamma_A+\mu\right)\\
&\hspace{2cm}-Nk\left(\beta_I\left(1-q\right)+\beta_A q\right)\biggr]\lambda+\left(k+\mu\right)\left(\gamma_I+\mu\right)\left(\gamma_A+\mu\right)\\
&\hspace{2cm}-Nk\biggl(\beta_I\left(1-q\right)\left(\gamma_A+\mu\right)+\beta_A q\left(\gamma_I+\mu\right)\biggr)=0.
\end{align*}

We prove the stability of this steady state using the Routh-Hurwitz criterion \cite{edelstein2005}. The disease-free steady state is stable if and only if 
\begin{align}
&\gamma_I+\gamma_A+k+3\mu>0,\label{rh_cond_1}\\
& \left(k+\mu\right)\left(\gamma_I+\mu\right)\left(\gamma_A+\mu\right)-\beta_I Nk\left(1-q\right)\left(\gamma_A+\mu\right)-\beta_A N kq\left(\gamma_I+\mu\right)>0,\label{rh_cond_2}\\
& \left(\gamma_I+\gamma_A+k+3\mu\right)\biggl[\left(k+\mu\right)\left(\gamma_I+\gamma_A+2\mu\right)+\left(\gamma_I+\mu\right)\left(\gamma_A+\mu\right)\nonumber\\
&\hspace{3cm}-Nk\biggl(\beta_I\left(1-q\right)+\beta_A q\biggr)\biggr]-\left(k+\mu\right)\left(\gamma_I+\mu\right)\left(\gamma_A+\mu\right)\nonumber\\
&\hspace{3cm}+Nk\biggl(\beta_I\left(1-q\right)\left(\gamma_A+\mu\right)+\beta_A q\left(\gamma_I+\mu\right)\biggr)>0.\label{rh_cond_3}
\end{align}
The inequality \eqref{rh_cond_1} always holds. The inequality \eqref{rh_cond_2} can be equivalently written as 
\begin{equation}
\mathcal{R}_0=\dfrac{\beta_I Nk\left(1-q\right)}{\left(k+\mu\right)\left(\gamma_I+\mu\right)}+\dfrac{\beta_A N kq}{\left(k+\mu\right)\left(\gamma_A+\mu\right)}<1.
\end{equation}
By using Mathematica \cite{Mathematica} we can confirm that the inequality \eqref{rh_cond_3} holds for $\mathcal{R}_0< 1$, thus the disease-free steady state is stable. 
\end{proof}

Since we incorporate the demographic terms, we are interested in exploring the longer-term persistence and the endemic dynamics of the disease. Setting equal to zero the right-hand side of system \eqref{SEIAR}, we find a unique endemic steady state.  Then, we are interested in determining the conditions necessary for endemic steady state stability. 

\begin{theorem}
If $\mathcal{R}_0> 1$, the endemic steady state, $\left(S^*,E^*,I^*,A^*,R^*\right)$, of system \eqref{SEIAR}  with
\begin{equation}
\begin{aligned}
        S^*&=\dfrac{\left(\gamma_I+\mu\right) \left(\gamma_A+\mu\right)\left(k+\mu\right)}{k\left(\beta_I\left(\gamma_A+\mu\right)+q\left(\beta_A\left(\gamma_I+\mu\right)-\beta_I\left(\gamma_A+\mu\right)\right)\right)}=\dfrac{N}{\mathcal{R}_0},\\
        E^*&=\dfrac{N\mu}{\left(k+\mu\right)}\left(1-\dfrac{1}{\mathcal{R}_0}\right), \\
        I^*&=\dfrac{Nk\left(1-q\right)\mu}{ \left(\gamma_I+\mu\right)\left(k+\mu\right)}\left(1-\dfrac{1}{\mathcal{R}_0}\right),\\
       A^*&=\dfrac{Nkq\mu}{ \left(\gamma_A+\mu\right)\left(k+\mu\right)}\left(1-\dfrac{1}{\mathcal{R}_0}\right),\\
        R^*&=\dfrac{Nk\biggl(q\gamma_A\mu+\gamma_I\left(\gamma_A+\mu-q\mu\right)\biggr)}{\left(\gamma_I+\mu\right) \left(\gamma_A+\mu\right)\left(k+\mu\right)}\left(1-\dfrac{1}{\mathcal{R}_0}\right),
        \label{endemic_ss}
\end{aligned}
\end{equation}
is locally stable.
\end{theorem}
\begin{proof}
The characteristic equation of the Jacobian matrix \eqref{jacobian} at the endemic steady state is
\begin{equation*}
    \lambda^4+\alpha_3\lambda^3+\alpha_2\lambda^2+\alpha_1\lambda+\alpha_0=0,
\end{equation*}
with
\begin{align*}
\alpha_3=&\gamma_I+\gamma_A+k+3\mu+\mu \mathcal{R}_0>0,\\
\alpha_2=&\mu \mathcal{R}_0\left(\gamma_I+\gamma_A+k+3\mu\right)+\left(\gamma_I+\gamma_A+2\mu\right)\left(k+\mu\right)+\left(\gamma_I+\mu\right)\left(\gamma_A+\mu\right)\\
&-\dfrac{N}{\mathcal{R}_0}k\left(\beta_I\left(1-q\right)+\beta_Aq\right),\\
\alpha_1=&\mu \mathcal{R}_0\left(\gamma_I+\gamma_A+2\mu\right)\left(k+\mu\right)+\left(\gamma_I+\mu\right)\left(\gamma_A+\mu\right)\left(\mu \mathcal{R}_0+k+\mu\right)\\
&-\dfrac{N}{\mathcal{R}_0}k\left(\beta_I\left(1-q\right)\left(\gamma_A+\mu+1\right)+\beta_Aq\left(\gamma_I+\mu+1\right)\right),\\
\alpha_0=&\mu \mathcal{R}_0\left(k+\mu\right)\left(\gamma_I+\mu\right)\left(\gamma_A+\mu\right)-\mu\dfrac{N}{\mathcal{R}_0}k\left(\beta_I\left(1-q\right)\left(\gamma_A+\mu\right)+\beta_Aq\left(\gamma_I+\mu\right)\right).
\end{align*}
From the Routh-Hurwitz criterion, the endemic steady \eqref{endemic_ss} is locally stable if and only if
\begin{equation*}
\alpha_0>0,\;\alpha_1>0,\;\alpha_3>0 \; \textnormal{and} \; \alpha_1\alpha_2\alpha_3-\alpha_1^2-\alpha_0\alpha_3^2>0.
\end{equation*}
We always have that $\alpha_3>0$, whereas $a_0>0$ is equivalent to $\mathcal{R}_0>1$.
By using Mathematica \cite{Mathematica} we can confirm that the rest of the above relations hold for $\mathcal{R}_0>1$, thus the endemic steady state is stable.
\end{proof}

\subsection{Global Stability Analysis of the SEIAR Model}

\begin{theorem}
If $\mathcal{R}_0\leq 1$, then the disease-free steady state, $\left(N,0,0,0,0\right)$, of system \eqref{SEIAR} is globally asymptotically stable. 
\label{global_stability_dfss}
\end{theorem}

\begin{proof}
We prove the global stability of the disease-free steady state $\left(N,0,0,0,0\right)$ by constructing a Lyapunov function. We consider the function $\mathcal{V}_1:  \mathbb{R^+}\times\mathbb{R}^3 \to\mathbb{R}$ with
\begin{equation*}
\mathcal{V}_1\left(S\left(t\right),E\left(t\right),I\left(t\right),A\left(t\right)\right)=\left(S-N-N\ln{\dfrac{S}{N}}\right)+E+\dfrac{N\beta_I}{\gamma_I+\mu} I+\dfrac{N\beta_A}{\gamma_A+\mu} A. 
\end{equation*}
We take the derivative of $\mathcal{V}_1$ with respect to $t$:
\begin{align*}
 \mathcal{V}_1^\prime&=S^\prime\left(1-\dfrac{N}{S}\right)+E^\prime+\dfrac{N\beta_I}{\gamma_I+\mu}I^\prime+\dfrac{N\beta_A}{\gamma_A+\mu}A^\prime\\
 &=2\mu N-\mu S-\dfrac{\mu N^2}{S}+E\left(k+\mu\right)\left(\dfrac{\beta_I Nk\left(1-q\right)}{\left(k+\mu\right)\left(\gamma_I+\mu\right)}+\dfrac{\beta_A N kq}{\left(k+\mu\right)\left(\gamma_A+\mu\right)}-1\right)\\
 &=-\mu N\left(\dfrac{S}{N}+\dfrac{N}{S}-2\right)+E\left(k+\mu\right)\left(\mathcal{R}_0-1\right).
\end{align*}
From the arithmetic–geometric mean inequality we have
\begin{equation*}
\dfrac{1}{2}\left(\dfrac{S}{N}+\dfrac{N}{S} \right)\geq \sqrt[2]{\dfrac{S}{N}\dfrac{N}{S}}=1\Rightarrow \dfrac{S}{N}+\dfrac{N}{S}-2\geq 0.    
\end{equation*}
Thus, if $\mathcal{R}_0\leq1$ then $\mathcal{V}_1^\prime\leq 0$  for all $t\geq 0$ and $\left(S,E,I,A\right)\in\mathbb{R^+}\times\mathbb{R}^3$ sufficiently close to  $\left(N,0,0,0\right)$, and $\mathcal{V}_1^\prime\left(t\right)=0$ holds only for $\left(S,E,I,A\right)=\left(N,0,0,0\right)$. Hence, the singleton $\{\left(N,0,0,0\right)\}$ is the largest invariant set for which $\mathcal{V}_1^\prime=0$. Then, from LaSalle's Invariance Principle \cite{lassalle1976} it follows that the disease-free steady state is globally asymptotically stable.
\end{proof}

\begin{theorem}
If $\mathcal{R}_0> 1$, then the endemic steady state, $\left(S^*,E^*,I^*,A^*,R^*\right)$, of system \eqref{SEIAR}  is globally asymptotically stable.
\label{global_stability_endss}
\end{theorem}

\begin{proof}
We consider the function $\mathcal{V}_2:\left( \mathbb{R^+} \right)^4 \to\mathbb{R}$ with
\begin{align*}
\mathcal{V}_2\left(S\left(t\right),E\left(t\right),I\left(t\right),A\left(t\right)\right)=&\left(S-S^*-S^*\ln{\dfrac{S}{S^*}}\right)+\left(E-E^*-E^*\ln{\dfrac{E}{E^*}}\right)\nonumber\\&+\dfrac{\beta_IS^*}{\gamma_I+\mu}\left(I-I^*-I^*\ln{\dfrac{I}{I^*}}\right)+\dfrac{\beta_AS^*}{\gamma_A+\mu}\left(A-A^*-A^*\ln{\dfrac{A}{A^*}}\right).
\end{align*}
We take the derivative of $\mathcal{V}_2$ with respect to $t$:
\begin{align*}
 \mathcal{V}_2^\prime=&S^\prime\left(1-\dfrac{S^*}{S}\right)+E^\prime\left(1-\dfrac{E^*}{E}\right)+\dfrac{\beta_IS^*}{\gamma_I+\mu}I^\prime\left(1-\dfrac{I^*}{I}\right)+\dfrac{\beta_AS^*}{\gamma_A+\mu}A^\prime\left(1-\dfrac{A^*}{A}\right)\\
 =&\left(\mu N-\beta_I SI-\beta_A SA-\mu S\right)\left(1-\dfrac{S^*}{S}\right)+\left(\beta_I SI+\beta_A SA -\left(k+\mu\right) E\right)\left(1-\dfrac{E^*}{E}\right)
 \\&+\dfrac{\beta_IS^*}{\gamma_I+\mu}\left(k \left(1-q\right)E-\left(\gamma_I+\mu\right)I\right)\left(1-\dfrac{I^*}{I}\right)+\dfrac{\beta_AS^*}{\gamma_A+\mu}\left(k q E-\left(\gamma_A+\mu\right)A\right)\left(1-\dfrac{A^*}{A}\right).
\end{align*}
After using the relations 
\[
\mu N=\beta_I S^*I^*+\beta_A S^*A^*+\mu S^*, 
\]
and
\[
\beta_I S^*I^*+\beta_A S^*A^* =\left(k+\mu\right) E^*, k\left(1-q\right)E^*=\left(\gamma_I+\mu\right)I^*, kqE^*=\left(\gamma_A+\mu\right)A^*,
\]
and adding and subtracting the terms $\dfrac{\beta_I S^*I^{*^2}E}{IE^*}$ and $\dfrac{\beta_A S^*A^{*^2}E}{AE^*}$ we have\\
\begin{align*}
 \mathcal{V}_2^\prime=&-\mu S^*\left(\dfrac{S}{S^*}+\dfrac{S^*}{S}-2\right) -\beta_I S^*I^*\left(\dfrac{S^*}{S}+\dfrac{S}{S^*}\dfrac{I}{I^*}\dfrac{E^*}{E}+\dfrac{I^*}{I}\dfrac{E}{E^*}-3\right)\\
 &-\beta_A S^*A^*\left(\dfrac{S^*}{S}+\dfrac{S}{S^*}\dfrac{A}{A^*}\dfrac{E^*}{E}+\dfrac{A^*}{A}\dfrac{E}{E^*}-3\right).
\end{align*}
From the arithmetic–geometric mean inequality we have that 
\begin{equation*}
 \dfrac{1}{3}\left(\dfrac{S^*}{S}+\dfrac{S}{S^*}\dfrac{I}{I^*}\dfrac{E^*}{E}+\dfrac{I^*}{I}\dfrac{E}{E^*} \right)\geq \sqrt[3]{\dfrac{S^*}{S}\dfrac{S}{S^*}\dfrac{I}{I^*}\dfrac{E^*}{E}\dfrac{I^*}{I}\dfrac{E}{E^*}}=1,
\end{equation*}
and
\begin{equation*}
 \dfrac{1}{3}\left(\dfrac{S^*}{S}+\dfrac{S}{S^*}\dfrac{A}{A^*}\dfrac{E^*}{E}+\dfrac{A^*}{A}\dfrac{E}{E^*} \right)\geq \sqrt[3]{\dfrac{S^*}{S}\dfrac{S}{S^*}\dfrac{A}{A^*}\dfrac{E^*}{E}\dfrac{A^*}{A}\dfrac{E}{E^*}}=1,
\end{equation*}
hence $\mathcal{V}_2^\prime\leq 0$ for all $\left(S,E,I,A\right)\in\left( \mathbb{R^+} \right)^4$, and the equality holds only for the endemic steady state  $\left(S^*,E^*,I^*,A^*\right)$. We conclude again from LaSalle's Invariance Principle that the endemic steady state is globally asymptotically stable.
\end{proof}
\subsection{Numerical Simulations for the SEIAR Model}
We proceed to the estimation of the already known epidemic curve of the disease in Italy, as obtained from the data set \cite{ECDC}, by $\dfrac{\mathrm{d}R}{\mathrm{d}t}$ (see, e.g., \cite{kermmck1927} and \cite{braun1993differential}). We plot together the two functions in Fig.~\ref{SEIAR_fit_data_italy}. The total population of Italy is 60,456,999. Once the restriction of movement (quarantine) during the manifestation of COVID-19 was applied, it limited the spread of the disease. To this end we follow the approach in \cite{Ndairou2020} and we consider as the total population N = 60,456,999/250.


\begin{figure}[h!]
        \centering
        \includegraphics[height=5cm]{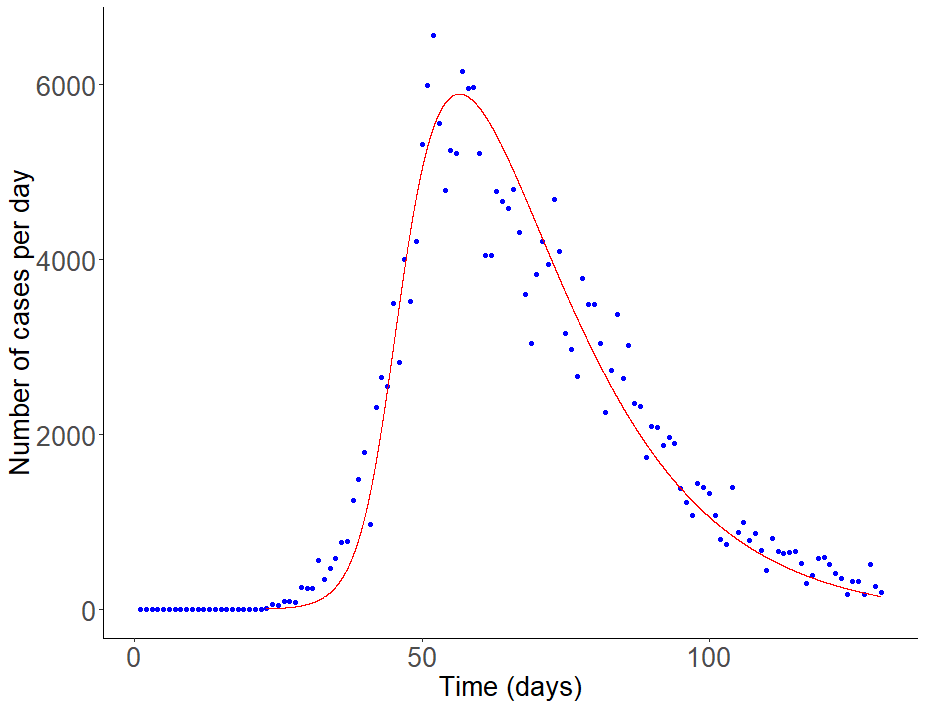}
       \caption{Number of confirmed cases per day in Italy until July 2020. The blue dots represents the data obtained from \cite{ECDC} and the red curve the graph of $\dfrac{\mathrm{d}R}{\mathrm{d}t}$, as obtained  by \eqref{SEIAR}. The parameters used here are: $\beta_I=2.55/N, \beta_A=1.275/N, k=0.07, \mu=0.001, \gamma_I=0.0625, \gamma_A=0.083, q=0.425$, and initial conditions (ICs): $S_0=N-200, I_0=100, A_0=100, E_0=R_0=0$.}
\label{SEIAR_fit_data_italy}      
\end{figure}

In Fig.~\ref{SEIAR_dynamics} we show the dynamics of the proportion of the values of model \eqref{SEIAR} for the set of parameters used in Fig.~\ref{SEIAR_fit_data_italy}.  We see in Fig.~\ref{SEIAR_dynamics;a}  that for this set of parameters the solution of the system has an oscillatory behaviour towards the endemic steady state. This can be more clear in Fig.~\ref{SEIAR_dynamics;b}, where the proportion of the infectious population (both symptomatic and asymptomatic) oscillates towards the proportion of the steady state $I^*+A^*$.

\begin{figure}[h!]
\begin{adjustwidth}{-.5in}{-.5in}  
 \centering
    \begin{subfigure}[b]{0.55\textwidth}
        \centering
        \includegraphics[height=5cm]{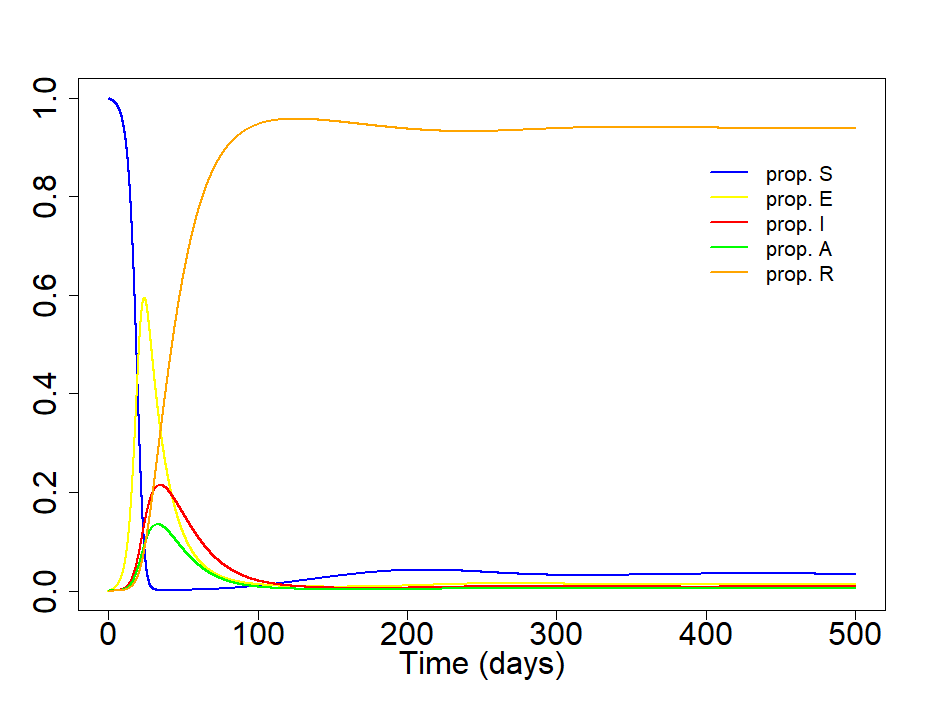}
        \caption{}
        \label{SEIAR_dynamics;a}  
    \end{subfigure}%
    ~ 
    \begin{subfigure}[b]{0.55\textwidth}
        \centering
        \includegraphics[height=5cm]{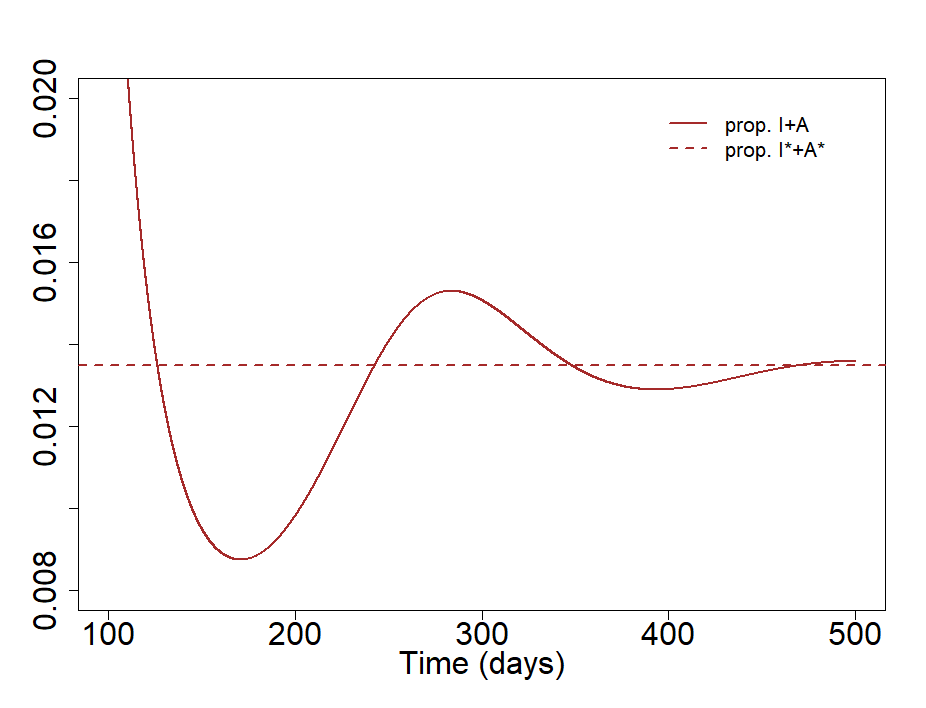}
        \caption{}
        \label{SEIAR_dynamics;b}  
    \end{subfigure}
    \end{adjustwidth}
  \caption{The dynamics of the proportion of the values of model \eqref{SEIAR}. (a) For all subgroups; (b) For the groups of both symptomatic and asymptomatic infectious individuals, $I+A$, towards the steady state $I^*+A^*$. The parameters used are the same as in Fig.~\ref{SEIAR_fit_data_italy} (see also Table~\ref{table_par}). The initial conditions are: $S_0=1-0.0008, I_0=0.0004, A_0=0.0004, E_0=R_0=0$.}
\label{SEIAR_dynamics}       
\end{figure}
\section{Modelling Transmission Dynamics of COVID-19 in a Vaccinated Population}
\label{modelling_vaccine}
In this section we consider the subclass of the vaccinated-with-a-prophylactic-vaccine ($V$) individuals. We set $0\leq p\leq 1$ for the vaccine coverage, as well as $0\leq \epsilon<1$ for the vaccine efficacy \cite{mclean1995}. Then, the model becomes
\begin{subequations}
\label{SVEIAR}
\begin{align}
\dfrac{\mathrm{d}S}{\mathrm{d}t}&=\left(1-p\right)\mu N-\beta_I SI-\beta_A SA  -\mu S,\label{SVEIAR;a}\\
\dfrac{\mathrm{d}V}{\mathrm{d}t}&=p\mu N-\left(1-\epsilon\right)\left(\beta_I VI+\beta_A VA\right)-\mu V,\label{SVEIAR;b}\\
\dfrac{\mathrm{d}E}{\mathrm{d}t}&=\left(S+\left(1-\epsilon\right)V\right)\left(\beta_I I+\beta_A A\right) -\left(k+\mu\right) E,\label{SVEIAR;c}\\
\dfrac{\mathrm{d}I}{\mathrm{d}t}&=k \left(1-q\right)E-\left(\gamma_I+\mu\right)I,\label{SVEIAR;d}\\
\dfrac{\mathrm{d}A}{\mathrm{d}t}&=k q E-\left(\gamma_A+\mu\right)A,\label{SVEIAR;e}\\
\dfrac{\mathrm{d}R}{\mathrm{d}t}&=\gamma_I I+\gamma_A A-\mu  R,\label{SVEIAR;f}
\end{align}
\end{subequations}
along with the initial conditions:
\begin{equation}
\left(S\left(0\right),V\left(0\right),E\left(0\right),I\left(0\right),A\left(0\right),R\left(0\right)\right)=\left(S_0,V_0,E_0,I_0,A_0,R_0\right)\in \left( \mathbb{R}_{0}^{+} \right)^6.
\end{equation}
Following the same steps as before and using the disease-free steady state of the model, $\left(\left(1-p\right)N, pN, 0, 0, 0, 0\right)$, we have that the basic reproductive ratio for the model where vaccination is applied is
\begin{equation}
\mathcal{R}_0^V=\left(1-\epsilon p\right)\biggl(\dfrac{\beta_I Nk\left(1-q\right)}{\left(k+\mu\right)\left(\gamma_I+\mu\right)}+\dfrac{\beta_A Nkq}{\left(k+\mu\right)\left(\gamma_A+\mu\right)}\biggr)=\left(1-\epsilon p\right)\mathcal{R}_0.
\label{R0V}
\end{equation}
The endemic steady state, $\left(S^*,V^*,E^*,I^*,A^*,R^*\right)$, of model \eqref{SVEIAR} is\\
\begin{align*}
&\Biggl(\dfrac{\left(1-p\right)N}{\mathcal{R}_0^V},\dfrac{pN}{\mathcal{R}_0^V},\dfrac{N\mu}{\left(k+\mu\right)}\left(1-\dfrac{1}{\mathcal{R}_0^V}\right),\dfrac{\left(1-\epsilon p\right)Nk\left(1-q\right)\mu}{ \left(\gamma_I+\mu\right)\left(k+\mu\right)}\left(1-\dfrac{1}{\mathcal{R}_0^V}\right),\nonumber\\& \dfrac{\left(1-\epsilon p\right)Nkq\mu}{ \left(\gamma_A+\mu\right)\left(k+\mu\right)}\left(1-\dfrac{1}{\mathcal{R}_0^V}\right), \dfrac{\left(1-\epsilon p\right)Nk\biggl(q\gamma_A\mu+\gamma_I\left(\gamma_A+\mu-q\mu\right)\biggr)}{\left(\gamma_I+\mu\right) \left(\gamma_A+\mu\right)\left(k+\mu\right)}\left(1-\dfrac{1}{\mathcal{R}_0^V}\right)\Biggr).
\end{align*}

\begin{figure}[h!]
\centering
\sidecaption[t]
\includegraphics[scale=0.95]{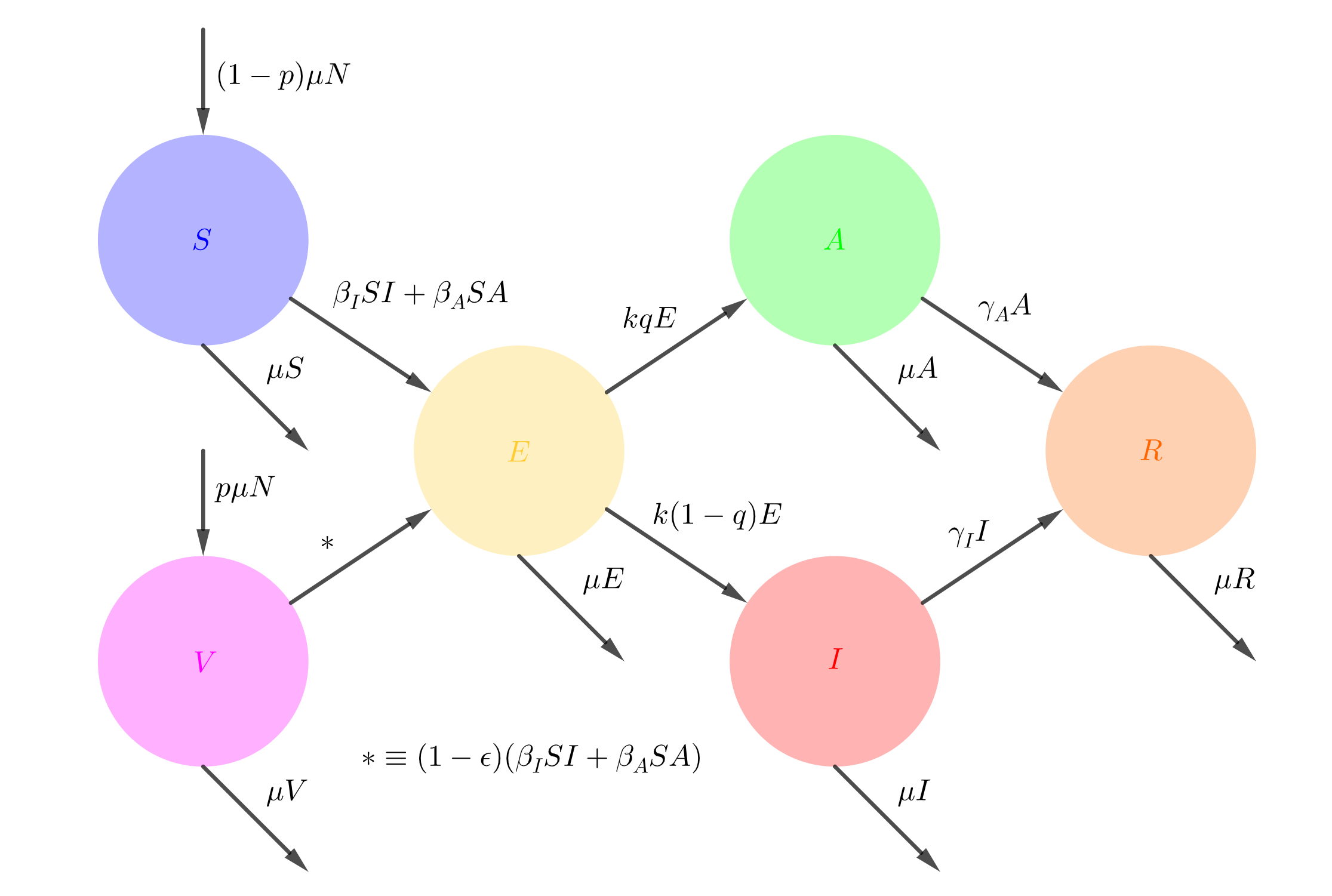}
\caption{Flow diagram of the SVEIAR model \eqref{SVEIAR}.}
\label{SVEIAR_flow_diagram}     
\end{figure}
We prove the global stability of the model following the same steps as before.
\begin{theorem}
If $\mathcal{R}_0^V\leq 1$, then the disease-free steady state, $\left(\left(1-p\right)N, pN, 0, 0, 0, 0\right)$, of system \eqref{SVEIAR} is globally asymptotically stable. 
\end{theorem}

\begin{proof}
We prove that the disease-free steady state of system \eqref{SVEIAR} is globally asymptotically stable by applying again the LaSalle's Invariance Principle for the Lyapunov function $\mathcal{V}_1^V: \left( \mathbb{R}^+ \right)^2\times\mathbb{R}^3 \to\mathbb{R}$ with
\begin{align*}
\mathcal{V}_1^V&\left(S\left(t\right),V\left(t\right),E\left(t\right),I\left(t\right),A\left(t\right)\right)=\left(S-\left(1-p\right)N-\left(1-p\right)N\ln{\dfrac{S}{\left(1-p\right)N}}\right)\nonumber\\
&+\left(V-pN-pN\ln{\dfrac{V}{pN}}\right)+E+\dfrac{\left(1-\epsilon p\right)N\beta_I}{\gamma_I+\mu} I+\dfrac{\left(1-\epsilon p\right)N\beta_A}{\gamma_A+\mu} A,
\end{align*}
and following the steps corresponding to the proof of Theorem~\ref{global_stability_dfss}.
\end{proof}

\begin{theorem}
If $\mathcal{R}_0^V> 1$, then the endemic steady state,  $\left(S^*,V^*,E^*,I^*,A^*,R^*\right)$, of system \eqref{SVEIAR} is globally asymptotically stable. 
\end{theorem}

\begin{proof}
We prove that the endemic steady state of system \eqref{SVEIAR} is globally asymptotically stable by applying again LaSalle's Invariance Principle for the Lyapunov function $\mathcal{V}_2^V: \left( \mathbb{R}^+ \right)^5 \to\mathbb{R}$, with
\begin{align*}
\mathcal{V}_2^V&\left(S\left(t\right),V\left(t\right),E\left(t\right),I\left(t\right),A\left(t\right)\right)=\left(S-S^*-S^*\ln{\dfrac{S}{S^*}}\right)+\left(V-V^*-V^*\ln{\dfrac{V}{V^*}}\right)\nonumber\\&+\left(E-E^*-E^*\ln{\dfrac{E}{E^*}}\right)+\dfrac{\beta_IS^*+\left(1-\epsilon\right)\beta_IV^*}{\gamma_I+\mu}\left(I-I^*-I^*\ln{\dfrac{I}{I^*}}\right)\nonumber\\&+\dfrac{\beta_AS^*+\left(1-\epsilon\right)\beta_AV^*}{\gamma_A+\mu}\left(A-A^*-A^*\ln{\dfrac{A}{A^*}}\right),
\end{align*}
and following the steps corresponding to the proof of Theorem~\ref{global_stability_endss}.
\end{proof}

\subsection{Numerical Simulations for the SVEIAR Model}
To assess vaccine effectiveness we focus on three important epidemiological measures \cite{Feng2011}: (i) the risk of infection spread, represented by $\mathcal{R}_0^V$; (ii) the peak prevalence of infection; (iii) the time at which the peak prevalence occurs. Relation \eqref{R0V} shows that the vaccine coverage, $p$, and vaccine efficacy, $\epsilon$, act multiplicatively on $\mathcal{R}_0$. As the proportion of asymptomatic cases is still unknown, in Fig~\ref{vaccine_effect_R0;a} we present a contour plot of the dependence of $\mathcal{R}_0^V$ on the vaccine coverage and vaccine efficacy, for different proportion of asymptomatic cases. The coloured curves represent the threshold $\mathcal{R}_0^V=1$ in \eqref{R0V}, between the infection spread (represented by the area below the threshold; $\mathcal{R}_0^V>1$) or not (represented by the area above the threshold; $\mathcal{R}_0^V<1$). The plot indicates that the vaccine efficacy and coverage need to be greater for small proportion of asymptomatic cases. As the number of symptomatic cases increases a more effective vaccine is needed. We see that even for a severe COVID-19 epidemic, as in our case, the vaccine can prevent the infection spread if both the vaccine efficacy and the vaccine coverage are high. Considering however that the data reflect a period where the severity of COVID-19 was not yet known and the average number of close contacts between individuals was very high due to occasions and events, the transmission rate, $\beta_I$, as obtained by the data is not the most appropriate index to predict vaccine effectiveness, as the situation has changed dramatically and close contacts have been significantly reduced. Hence, in Fig~\ref{vaccine_effect_R0;b} we present a corresponding contour plot for a lower $\beta_I$. We see that in the case of a reduced transmission rate the vaccine can prevent the infection spread, even for imperfect vaccines and small vaccine coverage.

In Fig~\ref{vaccine_effect_inf_dyn} we see the effect of vaccine efficacy on the proportion of the infection dynamics for high (Fig~\ref{vaccine_effect_inf_dyn;a}) and low (Fig~\ref{vaccine_effect_inf_dyn;b}) transmission rates. Higher vaccine efficacy leads to milder, but prolonged epidemics due to the slower rate of infection transmission. Moreover, it causes later occurrence of the first infection incidence and peak prevalence, and a slower rate of postpeak prevalence decline.

\begin{figure}[h!]
\begin{adjustwidth}{-.5in}{-.5in}  
 \centering
    \begin{subfigure}[b]{0.56\textwidth}
        \centering
        \includegraphics[height=5cm,width=5.6cm]{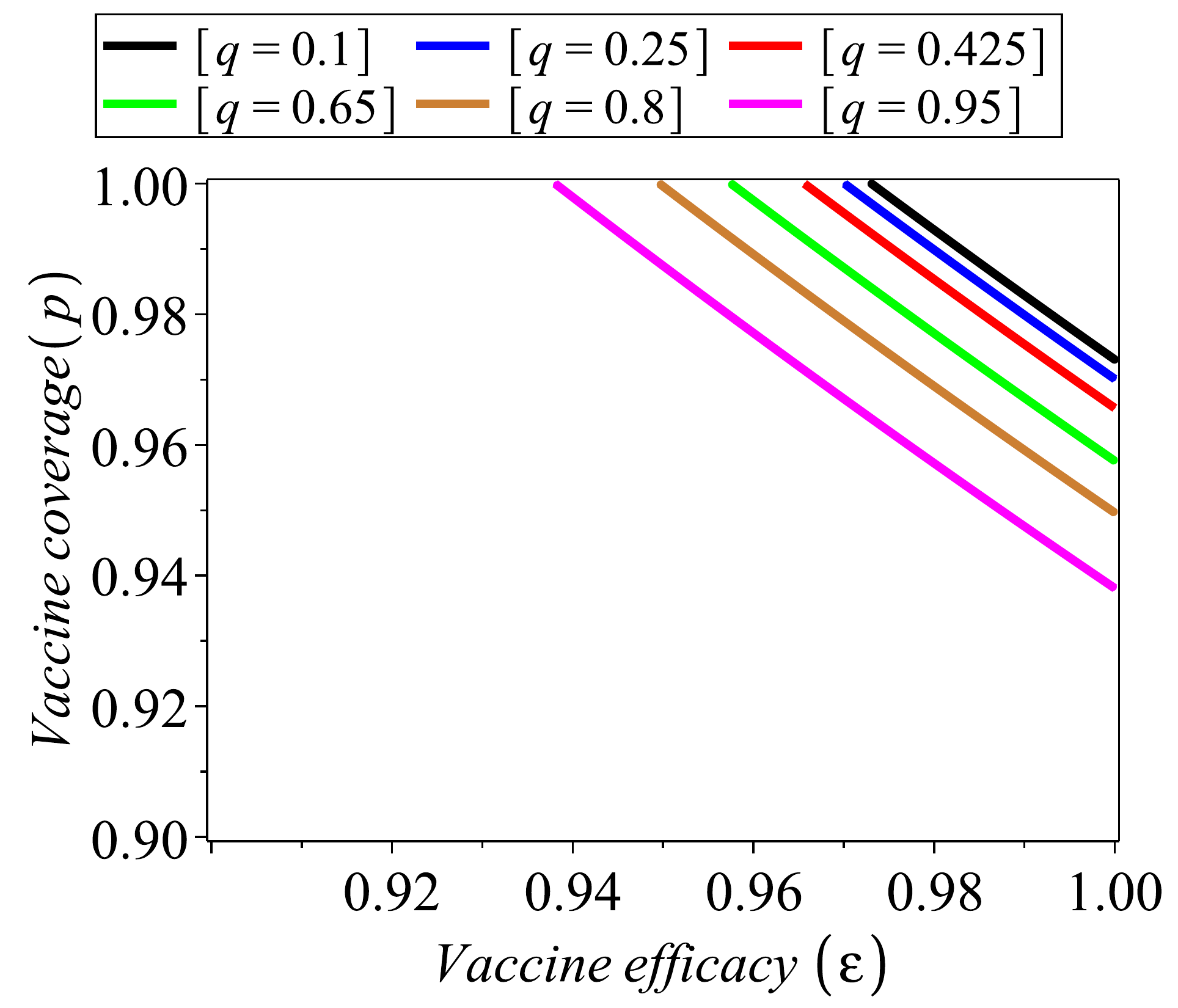}
        \caption{}
        \label{vaccine_effect_R0;a}  
    \end{subfigure}%
    ~ 
    \begin{subfigure}[b]{0.56\textwidth}
        \centering
        \includegraphics[height=5.58cm,,width=5.6cm]{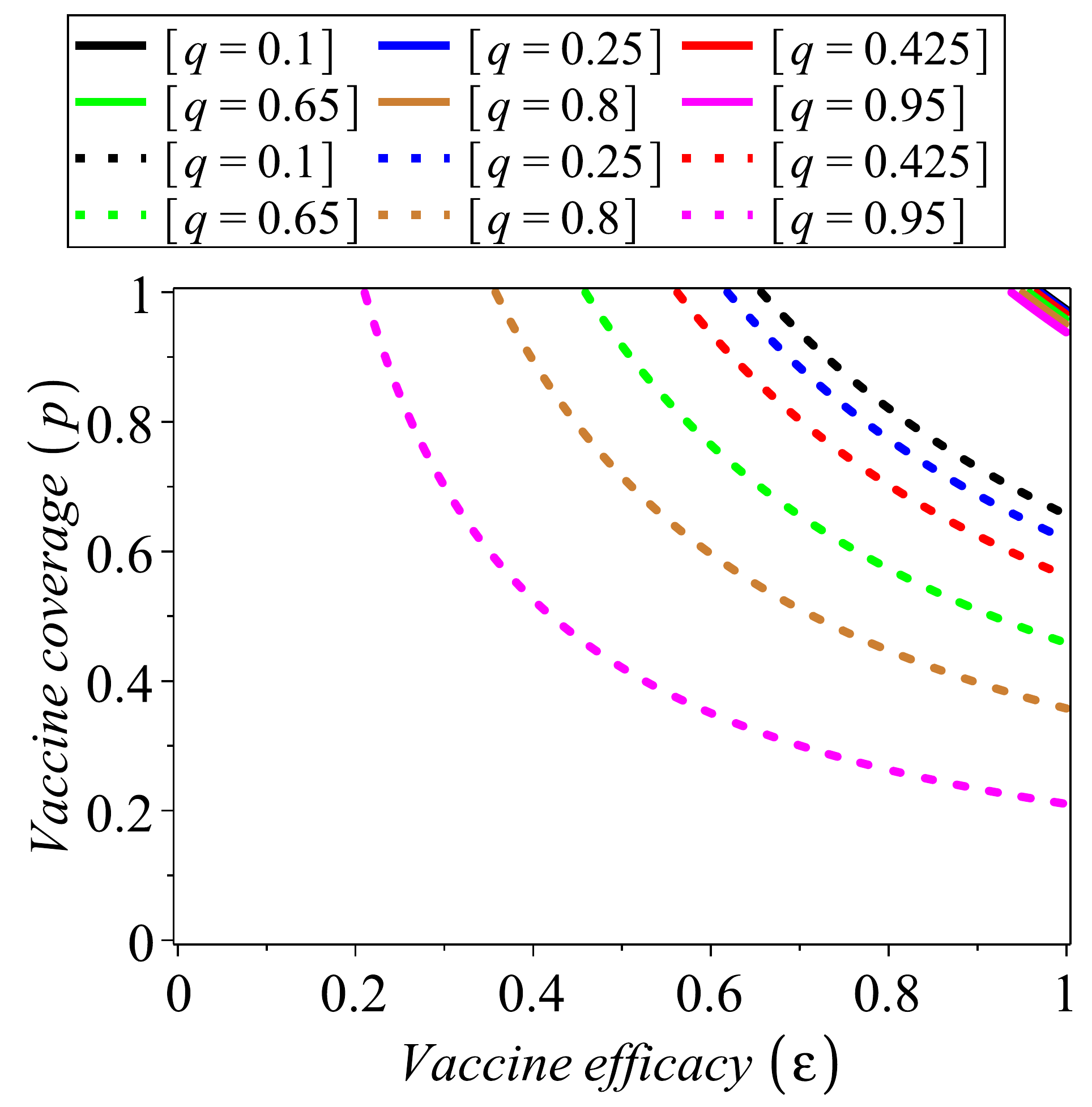}
        \caption{}
        \label{vaccine_effect_R0;b}  
    \end{subfigure}
      \end{adjustwidth}
   \caption{Assessing the vaccine effectiveness: A contour plot showing the dependence of $\mathcal{R}_0^V$ on the vaccine efficacy, vaccine coverage and proportion of asymptomatic cases for (a) $\beta_I=2.55$; (b) $\beta_I=0.2$. The rest of the model parameters are given in Table~\ref{table_par}.}  
\label{vaccine_effect_R0}       
\end{figure}

\begin{figure}[h!]
\begin{adjustwidth}{-.5in}{-.5in}  
 \centering
    \begin{subfigure}[b]{0.56\textwidth}
        \centering
        \includegraphics[height=5cm,width=5.6cm]{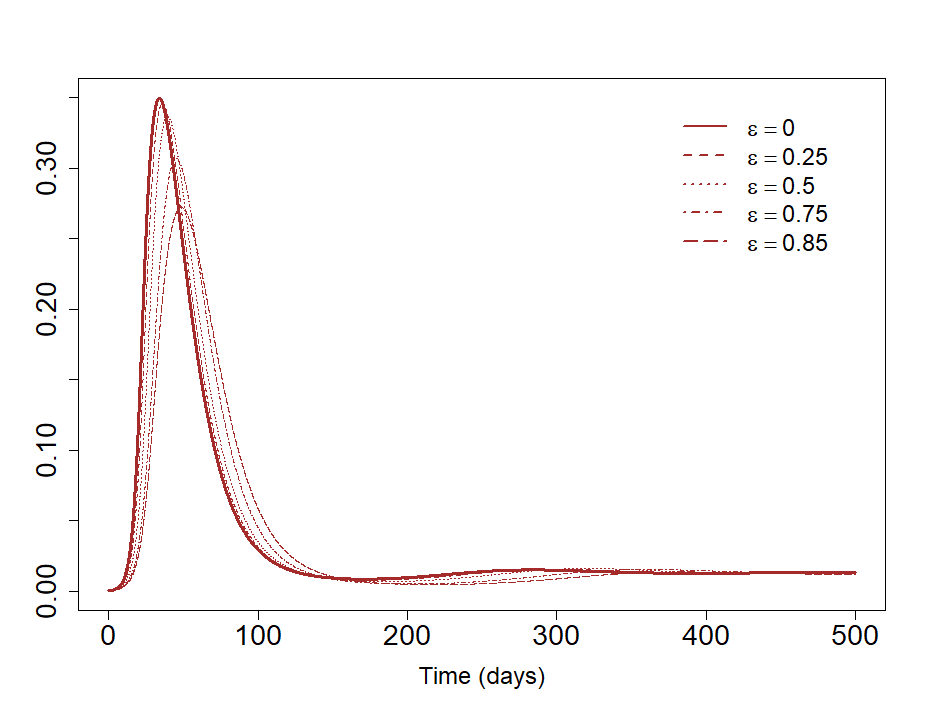}
        \caption{}
        \label{vaccine_effect_inf_dyn;a}  
    \end{subfigure}%
    ~ 
    \begin{subfigure}[b]{0.56\textwidth}
        \centering
        \includegraphics[height=5cm,,width=5.6cm]{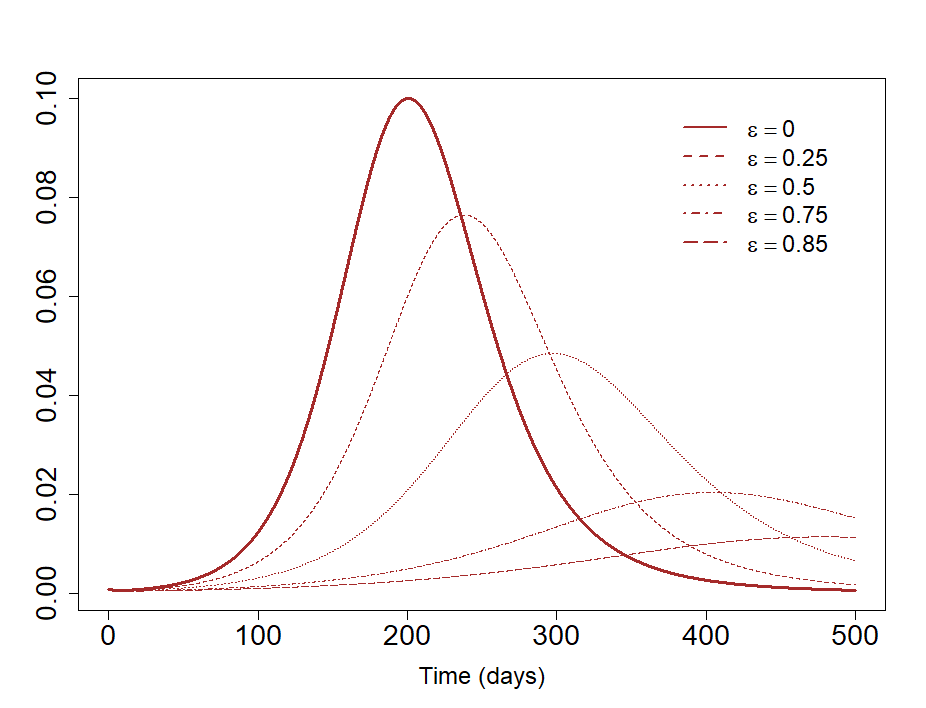}
        \caption{}
        \label{vaccine_effect_inf_dyn;b}  
    \end{subfigure}
      \end{adjustwidth}
 \caption{Assessing the vaccine effectiveness: The effect of vaccination on the prevalence of infection, with ICs: $S_0=V_0=0.4996, I_0=A_0=0.0004$, and  for (a) $\beta_I=2.55$; (b) $\beta_I=0.2$. The rest of the model parameters are given in Table~\ref{table_par}.}  
\label{vaccine_effect_inf_dyn}       
\end{figure}
\begin{table}[!t]
\begin{adjustwidth}{-.5in}{-.5in}  
        \begin{center}
\caption{Model parameters, values, units and relevant references.}
\label{model_param}      
\begin{tabular}{p{1cm}p{6cm}p{3.1cm}p{2.5cm}p{2cm}}
\hline\noalign{\smallskip}
Param. & Description  & Value & Unit  & Reference \\
\noalign{\smallskip}\svhline\noalign{\smallskip}
$\mu$ & Birth/Death rate & $0.001 (2\cdot10^{-5}-0.001)$  & days$^{-1}$ & \cite{UNdata,Keeling2011}\\
$\beta_I$ & Transmission rate of symptomatic infectious individuals & $2.55 (0.14-2.55)$ & individuals$^{-1}\cdot$days$^{-1}$ & \cite{Li2020,Ndairou2020,Pribylova2020,Castilho2020,Calafiore2020}\\
$\beta_A$ & Transmission rate of asymptomatic infectious individuals & $\dfrac{\beta_I}{2}$ & individuals$^{-1}\cdot$days$^{-1}$ & \cite{Sypsa2020}\\
$k$ & Incubation rate (rate of latent individuals becoming infectious) & $0.071 (0,071-0.2)$  &  days$^{-1}$ & \cite{Who202072,Lauer2020} \\
$q$ & Proportion of the asymptomatic infectious individuals & $0.425 (0-1)$  & - & \cite{Lavezzo2020}\\
$\gamma_I$ & Recovery rate of the symptomatic infectious individuals & $0.0625 (0.02-0.0625)$ & days$^{-1}$ & \cite{Yang2020,Zhou2020she,Zhou2020}\\
$\gamma_A$ & Recovery rate of the asymptomatic infectious individuals & $0.083 (0.083-0.33)$ & days$^{-1}$ & \cite{Yang2020,Zhou2020}\\
$p$ & Proportion of vaccinated individuals &  $0.5 (0-1)$ & - & Estimated\\
$\epsilon$ & Vaccine efficacy & $0.5 (0-1)$ & - & Estimated\\
\noalign{\smallskip}\hline\noalign{\smallskip}
\end{tabular}
\label{table_par}
\end{center}
    \end{adjustwidth}
\end{table}

\section{Conclusions}
\label{conclusions}
We presented an {\it ad hoc} SEIAR model with horizontal transmission and demographic terms for the epidemic spread of COVID-19, and we extended the model to include vaccination. The stability of both models is proved by implementing suitable  Lyapunov functions; the model is fitted  to real data from the epidemic in Italy. We studied the condition under which a vaccine can prevent disease spread. We accessed the vaccine effectiveness  focusing on the risk of infection spread, the peak prevalence of infection and the time at which the peak prevalence occurs.

Future work includes further investigation of the vaccine model, by incorporating different vaccination strategies, and if possible the comparison with biological data. An extension of the model will also include additional important factors of  COVID-19 spread, such as the population age, the geographical spread of the the epidemics (see e.g. Refs \cite{Diekmann1978,Khachatryan2020,Sergeev2019} and other references therein) and the waning immunity gained by infected individuals, as well as vertical transmission and migration terms for the infected individuals.



\begin{thebibliography}{99.}%
%
%
 \bibitem{gorbalenya2020} Gorbalenya, A. E., Baker, S. C., Baric, R. S., de Groot, R. J., Drosten, C., Gulyaeva, A. A., et al.: The species Severe acute respiratory syndrome-related coronavirus: classifying 2019-nCoV and naming it SARS-CoV-2. Nature Microbiology. \textbf{5} (4), 536-–544 (2020)

\bibitem{Worldometer}  Worldometer - www.worldometers.info

\bibitem{Lavezzo2020}   Lavezzo, E. et al.: Suppression of a {SARS}-{C}o{V}-2 outbreak in the Italian municipality of Vo’. Nature (2020)

\bibitem{Yang2020} Yang, R., Gui, X., \& Xiong, Y.: Comparison of clinical characteristics of patients with asymptomatic vs symptomatic coronavirus disease 2019 in Wuhan, China. JAMA Network Open, \textbf{3} (5), e2010182-e2010182 (2020)

\bibitem{Li2020} Li, R., Pei, S., Chen, B., Song, Y., Zhang, T., Yang, W., \& Shaman, J.: Substantial undocumented infection facilitates the rapid dissemination of novel coronavirus (SARS-CoV-2). Science, \textbf{368} (6490), 489--493 (2020)


\bibitem{Heneghan2020} Heneghan, C., Brassey, J., \& Jefferson, T.: COVID-19: What proportion are asymptomatic? CEBM (2020)

\bibitem{Diekmann1990} Diekmann, O., Heesterbeek, J. A. P., Metz, J. A.: On the definition and the computation of the basic reproduction ratio $R_0$ in models for infectious diseases in heterogeneous populations. J. Math. Biol. \textbf{28} (4), 365--382 (1990) 
 
\bibitem{Diekmann2010} Diekmann, O., Heesterbeek, J. A. P., Roberts, M. G.: The construction of next-generation matrices for compartmental epidemic models. J. R. Soc. Interface. \textbf{7} (47), 873--885  (2010)

\bibitem{edelstein2005}  Edelstein-Keshet, L.: Mathematical Models in Biology. SIAM (2005)
 
\bibitem{Mathematica} Wolfram Research, Inc., Mathematica, Version 12.1, Champaign, IL (2020).


\bibitem{lassalle1976} La Salle, J. P.: The Stability of Dynamical Systems. SIAM (1976)


\bibitem{ECDC} European Centre for Disease Prevention and Control (ECDC) (2020) https://www.ecdc.europa.eu/en/publications-data/download-todays-data-geographic-distribution-covid-19-cases-worldwide


\bibitem{kermmck1927} Kermack, W. O., \& McKendrick, A. G.: A contribution to the mathematical theory of epidemics. Proc. Roy. Soc. Lond. Math. Phys. Sci., \textbf{115} (772), 700--721(1927)

\bibitem{braun1993differential} Braun, M.: Differential Equations and Their Applications, 4th ed. New York: Springer-Verlag (1993)

\bibitem{Ndairou2020} Ndairou, F., Area, I., Nieto, J. J., \& Torres, D. F.: Mathematical modeling of COVID-19 transmission dynamics with a case study of Wuhan. Chaos, Solitons \& Fractals, 109846 (2020)


\bibitem{mclean1995} McLean, A. R.: Vaccination, evolution and changes in the efficacy of vaccines: a theoretical framework. Proc. Biol. Sci. \textbf{261} (1362), 389--393 (1995)

\bibitem{Feng2011} Feng, Z., Towers, S., \& Yang, Y.: Modeling the effects of vaccination and treatment on pandemic influenza. AAPS J., \textbf{13} (3), 427--437  (2011)

\bibitem{UNdata} UNdata: Crude birth/death rate (per 1,000 population). United Nations (2020)

\bibitem{Keeling2011} Keeling, M. J., \& Rohani, P.: Modeling infectious diseases in humans and animals. Princeton University Press (2011)

\bibitem{Pribylova2020} Pribylova, L., \& Hajnova, V.: SEIAR model with asymptomatic cohort and consequences to efficiency of quarantine government measures in COVID-19 epidemic. arXiv:2004.02601 (2020)

\bibitem{Castilho2020} Castilho, C., Gondim, J. A., Marchesin, M., \& Sabeti, M.: Assessing the efficiency of different control strategies for the COVID-19 epidemic. EJDE, \textbf{2020} (64), 1--17 (2020)

\bibitem{Calafiore2020} Calafiore, G. C., Novara, C., \& Possieri, C.: A modified {SIR} model for the {COVID}-19 contagion in Italy. arXiv:2003.14391  (2020)


\bibitem{Sypsa2020} Sypsa, V., Roussos, S., Paraskevis, D., Lytras, T., Tsiodras, S., \& Hatzakis, A.: Modelling the SARS-CoV-2 first epidemic wave in Greece: social contact patterns for impact assessment and an exit strategy from social distancing measures. medRxiv  (2020)

\bibitem{Who202072} World Health Organization: Coronavirus disease 2019 (COVID-19): situation report, 72 (2020)

\bibitem{Lauer2020} Lauer, S. A., et al.: The incubation period of coronavirus disease 2019 (COVID-19) from publicly reported confirmed cases: estimation and application. Ann. Intern. Med. 172.9, 577-582 (2020)

\bibitem{Zhou2020she} Zhou, B., She, J., Wang, Y., \& Ma, X.: The duration of viral shedding of discharged patients with severe COVID-19. Clin. Infect. Dis.  (2020). 

\bibitem{Zhou2020} Zhou, R., Li, F., Chen, F., Liu, H., Zheng, J., Lei, C., \& Wu, X.: Viral dynamics in asymptomatic patients with {COVID}-19.  Int. J. Infect. Dis. 96:288--290 (2020)

\bibitem{Diekmann1978} Diekmann, O.: Thresholds and travelling waves for the geographical spread of infection. J. Math. Biol., \textbf{6} (2), 109 (1978)

\bibitem{Khachatryan2020} Khachatryan, K. A., Narimanyan, A. Z., \& Khachatryan, A. K.: On mathematical modelling of temporal spatial spread of epidemics. Math. Model. Nat. Phenom., \textbf{15} (6), 1--14 (2020)

\bibitem{Sergeev2019} Sergeev, A., \& Khachatryan, K.:  On the solvability of a class of nonlinear integral equations in the problem of a spread of an epidemic. Trans. Moscow Math. Soc., 80, 95--111 (2019)

\end{thebibliography}
\end{document}